\documentclass[a4paper]{amsart}

\usepackage{amsmath, amsfonts, amsthm, framed}
\usepackage{xspace}
\usepackage{newunicodechar}
\usepackage{fullpage}
\usepackage{enumitem}
\usepackage{quiver}
\usepackage{fontawesome}
\usepackage[numbers]{natbib}
\usepackage{soul}
\usepackage{cancel}
\usepackage{xspace}
\usepackage{eucal}
\usepackage{colonequals}
\usepackage[colorlinks=true]{hyperref}
\definecolor{darkergreen}{rgb}{0.0, 0.5, 0.0}
\definecolor{Blue}{RGB}{0,0,148}

\usepackage{listings}

\definecolor{light-gray}{gray}{0.95}
\usepackage{tikz}
\usetikzlibrary{positioning, arrows.meta, shapes.geometric}
\usepackage[a4paper, margin=0.75in]{geometry} 

\definecolor{keywordcolor}{rgb}{0.7, 0.1, 0.1}   
\definecolor{commentcolor}{rgb}{0.4, 0.4, 0.4}   
\definecolor{symbolcolor}{rgb}{0, 0, 0.8}    
\definecolor{tacticcolor}{rgb}{0, 0, 0.8}    
\definecolor{sortcolor}{rgb}{0.1, 0.5, 0.1}      
\newcommand*{\lean}[1]{\lstinline{#1}\xspace} 
\theoremstyle{plain}
\newtheorem{theorem}{Theorem}[section]

\newtheorem{lemma}[theorem]{Lemma}

\theoremstyle{definition}
\newtheorem{definition}[theorem]{Definition}
\newtheorem{remark}[theorem]{Remark}

\newcommand*{\N}{\mathbb{N}}
\newcommand*{\Z}{\mathbb{Z}}
\newcommand*{\Q}{\mathbb{Q}}
\newcommand*{\C}{\mathbb{C}}

\newcommand*{\OK}{\mathcal{O}_K}
\newcommand*{\Ok}{\mathcal{O}_k}
\newcommand*{\OL}{\mathcal{O}_L}
\newcommand{\gothp}{\mathfrak{p}}

\newcommand{\gothz}{\mathfrak{z}}

\DeclareMathOperator{\Tr}{Tr}

\DeclareMathOperator{\Norm}{Norm}
\DeclareMathOperator{\Gal}{Gal}
\DeclareMathOperator{\tors}{tors}
\DeclareMathOperator{\rk}{rk}

\newcommand*{\mathlib}{\textsc{mathlib}\xspace} 
\newcommand*{\Lean}{\textsc{Lean}\xspace} 
\newcommand*{\Leant}{\textsc{Lean3}\xspace}
\newcommand*{\Leanf}{\textsc{Lean4}\xspace}
\newcommand*{\fltregular}{\textsc{flt-regular}\xspace}
\newcommand*{\e}{\mathrm{e}}
\title{A complete formalization of Fermat's Last Theorem for regular primes in Lean}

\author{Alex J. Best}
\urladdr{\href{https://alexjbest.github.io/}{https://alexjbest.github.io/}}
\email{alexjbest@gmail.com}
\author{Christopher Birkbeck}
\address{University of East Anglia}
\email{c.birkbeck@uea.ac.uk}
\urladdr{\href{https://cdbirkbeck.wixsite.com/website}{https://cdbirkbeck.wixsite.com/website}}
\author{Riccardo Brasca}
\address{Université Paris Cité and Sorbonne Université, CNRS, IMJ-PRG, F-75013 Paris, France.}
\email{riccardo.brasca@imj-prg.fr}
\urladdr{\href{https://webusers.imj-prg.fr/~riccardo.brasca/}{https://webusers.imj-prg.fr/~riccardo.brasca/}}
\author{Eric Rodriguez Boidi}
\email{eric.rodriguez-boidi@kcl.ac.uk}
\author{Ruben Van de Velde}
\email{ruben.vandevelde@gmail.com}
\author{Andrew Yang}
\email{a.yang24@imperial.ac.uk}

\begin{document}

\begin{abstract}
We formalize a complete proof of the regular case of Fermat's Last Theorem in the \Leanf{} theorem prover. Our formalization includes a proof of Kummer's lemma, which is the main obstruction to Fermat's Last Theorem for regular primes. Rather than following the modern proof of Kummer's lemma via class field theory, we prove it by using Hilbert's Theorems 90-94 in a way that is more amenable to formalization. 
\end{abstract}

\maketitle

\section{Introduction}
For $x,y,z \in \Z$ and $n \in \N$ with $n> 2$, Fermat's Last Theorem (FLT) is the result that there are no solutions to
\[
x^n+y^n=z^n
\]
with $x$, $y$ and $z$ all different from $0$. This apparently easy result has been stated by Pierre de Fermat around 1637, but the first full proof, by Andrew Wiles and Richard Taylor, only appeared in 1995 in the two papers~\cite{wiles, taylor-wiles}. The quest for a proof of FLT has a long history of driving the development of new mathematics and this has continued now as mathematicians and computer scientists look to formalize mathematics. The importance of FLT lies not in the theorem itself (as number theory is full of seemingly easy equations that turn out to be extremely difficult to solve), but in the theories that have been developed to solve it. Indeed, attempting to prove FLT played a major role in the development of algebraic number theory.

Throughout the paper, we use the symbol \faExternalLink\xspace for external links. Almost every mathematical statement and definition will be accompanied by such a link directly to the source code for the corresponding statement in \mathlib{} or in \fltregular. To keep the links usable, they are all to a fixed commit of the master branch (the most recent one at the time of writing).

A routine argument shows that it is enough to prove FLT in the case where the exponent is equal to $4$ or it is a prime number $p \neq 2$. The case $n = 4$, proved by Fermat himself, is not completely obvious but was already in \mathlib{} at the beginning of our project, see \lean{fermatLastTheoremFour}\href{https://github.com/leanprover-community/mathlib4/blob/4705244150fe57b4006d3393c5040d3d74e024b4/Mathlib/NumberTheory/FLT/Four.lean#L298}{\faExternalLink}. In what follows we will describe the formalization (in the \Leanf{} theorem prover) of FLT in the special case where the exponent $p$ is what is known as a regular prime. That is, $p$ is a prime number (different from $2$) that does not divide the class number of the $p$-th cyclotomic extension $\Q(\e^{\frac{2\pi i}{p}})/\Q$. This case is significantly simpler than the full version and it is amenable to formalization given the current state of our formalized libraries, such as \mathlib{}~\cite{mathlib}.\footnote{A partial formalization of the full proof is a \href{https://github.com/ImperialCollegeLondon/FLT}{work in progress} by Kevin Buzzard, and it will likely take several years to complete.} There are several reasons for wanting to formalize the regular case. Firstly, it is a stress test of \mathlib{}: even though the tools we need here are much more basic than those required for the full proof, by formalizing this case we have extensively improved the number theory library, finding various issues with \mathlib{}'s original implementation. As a result of this work, those issues have now been addressed. Secondly, many of the tools required are of independent interest, such as discriminants, cyclotomic fields, ramification results, etc., all of which were originally formalized as a part of this work. Having \fltregular as our goal allowed us to develop these theories in a cohesive manner. Lastly, this represents the first formalization of a non-trivial family of cases of FLT (a family that is conjecturally infinite).

In the previous work~\cite{case1} of (some of) the authors, we formalized the following first step towards proving FLT in the regular case:
\begin{theorem}[Case 1]
	Let $p$ be an odd regular prime. Then $x^p+y^p=z^p$ has no solutions with $x,y,z \in \Z$ and $\gcd(xyz,p)=1$.
\end{theorem}
Here we will be concerned with the full proof (of the regular case), which replaces the condition that $\gcd(xyz,p)=1$ with $xyz \ne 0$. Specifically we will prove what is known as the second case, where the main difficulty is the need for Kummer's Lemma, which states that:
\begin{lemma}
Let $p$ be an odd regular prime and let $\zeta \in \C$ be a primitive $p$-th root of unity. If $u \in \Z[\zeta]^\times$ is a unit such that $u \equiv a \bmod p$ for an integer $a$, then there exists $v \in \Z[\zeta]^\times$ such that $u = v^p$.
\end{lemma}
While there are several proofs of this result, many make use of class field theory or the  $p$-adic class number formula, which are beyond what is currently available in \mathlib. For this reason, we have instead opted for a more `elementary' proof, requiring only some basic results about group cohomology, specifically Hilbert's Theorems 90-94.

Here is an outline of the paper. In Section~\ref{sec: notation} we fix the informal and formal notation that will be used throughout the paper, explaining the very first issues that appear in the formalization process. In Section~\ref{sec: case 1}, we recall how in~\cite{case1} we formalized Theorem~\ref{thm: case 1}, reducing FLT to Theorem~\ref{thm: case 2}. The main inputs for the proof are described in the following diagram:

\begin{figure}[htbp]
\centering
\begin{tikzpicture}[
    node distance=1.8cm and 1.5cm, 
    every node/.style={
        draw,
        rectangle,
        thick,
        align=center,
        minimum height=1cm,
        text width=4cm 
    },
    main_topic/.style={ 
        fill=blue!10,
        font=\bfseries
    },
    support_topic/.style={ 
        fill=green!10
    },
    arrow/.style={-Stealth, thick, rounded corners=2pt}
]


\node[main_topic]    (B)              {\textbf{Hilbert 91}};
\node[main_topic]    (C) [below=of B] {\textbf{Hilbert 92 (Thm 7.1)}};
\node[main_topic]    (D) [below=of C] {\textbf{Hilbert 94 (Thm 6.1 + 6.2)}};
\node[main_topic]    (E) [below=of D] {\textbf{Kummer's Lemma (Thm 5.1)}};
\node[main_topic]    (J) [below=of E] {\textbf{FLT for regular primes (Case II)}};

\node[main_topic] (K) [right=of C, xshift=-1cm] {\textbf{Hilbert 90 (Thm 8.1)}};
\node[support_topic] (F) [left=of C, xshift=1cm] {\textbf{Ramification of infinite places}};

\node[support_topic] (H) [left=of D, xshift=1cm] {\textbf{Unramified Extensions}};

\node[support_topic] (G) [right=of E, xshift=1cm, yshift=0.7cm] {\textbf{Different Ideal \& Ramification}};
\node[support_topic] (I) [right=of E, xshift=1cm, yshift=-0.7cm] {\textbf{Kummer Extensions}};


\draw[arrow] (B) -- (C);
\draw[arrow] (C) -- (D);
\draw[arrow] (D) -- (E);
\draw[arrow] (E) -- (J);
\draw[arrow] (K.west) -- (C.east);

\draw[arrow] (F.east) -- (C.west);
\draw[arrow] (H.east) -- (D.west);

\draw[arrow] (G.west) -- (E.east);
\draw[arrow] (I.west) -- (E.east);

\end{tikzpicture}
\caption{Main inputs into the proof.}
\label{fig:mermaid_to_tikz}
\end{figure}

The formalization process works backwards from the final goal, gradually reducing the results to known theorems. As we will see, working backwards is key since it ensures all our results will combine to prove our theorem. For this reason we follow a similar structure here, where results are introduced when needed and we gradually reduce our proofs, firstly to Kummer's lemma, then to Hilbert 94, 92,91, and 90 in that order.

Specifically, in Section~\ref{sec: case 2} we reduce the proof of Case 2 to Kummer's lemma, Theorem~\ref{thm: kummer}. This is mathematically not very difficult, but the proof is rather intricate, and its formalization is quite challenging. We then move on to the heart of our work, the formalization of Kummer's lemma. In Section~\ref{sec: kummer}, we prove Kummer's lemma \emph{assuming Hilbert's Theorem 94}, a cohomological result that is historically a precursor of global class field theory. In particular, the only thing that remains to be proven is Theorem~\ref{thm: Hilbert 94 2}: if $L/K$ is an unramified extension of number fields of odd prime degree, then $[L:K]$ divides the class number of $K$. 

Starting with Section~\ref{sec: Hilbert94}, our work is basically independent of FLT, and only concerns the cohomology of number fields. We then reduce Theorem~\ref{thm: Hilbert 94 2} to Theorem~\ref{thm: Hilbert 94 1}, and the latter to Hilbert's Theorems 92 and 90, Theorems~\ref{thm: Hilbert 92} and~\ref{thm: Hilbert 90 concr}. In Section~\ref{sec: Hilbert92} we prove Hilbert's Theorem 92, which is the hardest part of our formalization. The proof we formalize is a simplification of Hilbert's original proof, but it is still subtle and requires a lot of care. In particular it depends on nontrivial results about the units of the ring of integers of number fields. In Section~\ref{sec: Hilbert90} we explain the formalization of the version of Hilbert's Theorem 90 that we need, using what is already in \mathlib. This finishes the proof of Kummer's lemma, and hence the formalization of FLT for regular primes. In Section~\ref{sec: form stuff} we discuss some of the issues we encountered in the formalization process, and how we solved them. Finally, in Section~\ref{sec: future} we discuss some future work.

As is clear from the above discussion, the paper works backwards from the final result we want to prove to the results we need. This reflects our formalization process: the final goal, FLT for regular primes, was clear from the beginning and the statement was very easy to formalize. Also the fact that we needed to split the proof into two cases was clear, as all the known approaches follow this strategy. Once Case 1 was done, we moved to Case 2, and again all known approaches use Kummer's lemma. A key observation is that, even if the proof is difficult, the formalization of the statement is easy, so we decided to finish the work assuming Kummer's lemma. Now, as there are several possible ways of proving Theorem~\ref{thm: kummer} that involve various complicated objects (cohomology groups, $L$-functions etc.) it was key to choose an approach that could be formalized in a reasonable time-frame. 

Working backwards, we were sure that we were really making progress at any step of the formalization (of course there is always the risk of ending up with something very hard to prove, but in this way we minimize the risk of formalizing unneeded  \emph{definitions}). This is a very different strategy with respect to projects like the Liquid Tensor Experiment \cite{LTE} where the \emph{statement} of the result is already difficult to formalize, not only its proof.
\subsection*{What we learned:}
We found that going back to early proofs of theorems and recasting the main ideas with relatively more modern tools, can be highly beneficial for formalization, since these require much less advanced machinery, which in turn simplifies the formalization process. In particular, while Hilbert's theorem 90 is very well-known, theorems 91,92, 94 (of \cite{Hilbert}) were unknown to the authors before this project, but they are the key to our approach.

The project also uncovered several issues with Lean's \mathlib{}, for example, the definition of rings of integers has since changed, due in part to this project highlighting the fact that they would not scale well (as described in Section 9). In particular, it demonstrates the need to test definitions by proving non-trivial theorems.
\subsection*{Acknowledgements} Our project has greatly benefited from the contributions of many members of the \mathlib{} community. We are deeply grateful to all those who have been involved in the development of both \Lean and \mathlib{}, whose efforts have made our work possible. We would also like to extend our sincere thanks to the anonymous referees for their thorough and thoughtful review of an earlier version of this paper, which led to significant improvements in its clarity and presentation.

\section{Informal and formal notation} \label{sec: notation}
Throughout the paper, $p$ will be an odd regular prime. Although FLT is a statement about the integers, the proofs all require significant input from algebraic number theory. Here is our basic setup. Let $\mu_p \subset \C$ be the set of $p$-th roots of unity. We fix $\zeta \in \mu_p$ a primitive root of unity and we will write $K$ for the number field $\Q(\zeta)$. By the definition of regularity, $p$ does not divide the class number of $K$. The ring of integers of $K$ will be denoted $\OK$ and note that we have $\OK = \Z[\zeta]$ (a fact that we formalized in~\cite{case1}).

Even if our main result concerns natural numbers and the field $\Q(\zeta)$ only plays an auxiliary role, it appears in most of the results and so we need a formalized definition of $\Q(\zeta)$ that works well in practice. As is customary in formalization projects (see for example \cite[Section 4.3]{schemes}), it is better to avoid working directly with the field $\Q(\zeta)$ (that would be \lean{Algebra.adjoin ℚ \{ζ\}}), otherwise all our results would apply only to that extension of $\Q$, and not to any extension that is abstractly isomorphic to it. Instead, we work with any $p$-th cyclotomic extension of $\Q$, as follows.
\begin{lstlisting}
    variable {p : ℕ+} {K : Type*} [Field K] [NumberField K] [IsCyclotomicExtension {p} ℚ K]
\end{lstlisting}
The class \lean{IsCyclotomicExtension}\href{https://github.com/leanprover-community/mathlib4/blob/4705244150fe57b4006d3393c5040d3d74e024b4/Mathlib/NumberTheory/Cyclotomic/Basic.lean#L67-L78}{\faExternalLink}, that is a fundamental prerequisite for our formalization, has
been introduced into \mathlib{} during the first steps of our project (see~\cite{case1}). The instances
\lean{[Field K]} and \lean{[NumberField K]} ensure that $K$ is an extension of $\Q$ that is generated by a primitive $p$-th root of unity. Note that the assumption that $K$ is a number field is redundant, as it is implied by $K$ being a $p$-th cyclotomic extension of $\Q$: this result is easily proved in \Lean{}, but it cannot be an instance (since \Lean{} has no way of guessing $p$ when it sees a goal of type \lean{NumberField K}) and so \lean{NumberField K} is not found automatically by \Lean{} and in practice it is more convenient to assume it explicitly. For the same reason, we sometimes use \lean{[CharZero K]} instead of \lean{[NumberField K]} (since some results in \mathlib{} use this formulation). Note that in this project $p$ is of type \lean{(p : ℕ+)}, which is the type of positive natural numbers, and in particular it cannot be used where a natural number is expected (for example one cannot write \lean{x ^ p} directly): this is a rather standard annoyance caused by type theory and \Lean{} has a powerful system to help ignoring these issues (see for example \cite{coercions} for more details). Despite the coercions system, we are sometimes forced to use the coercion \lean{(↑p : ℕ)} explicitly to see $p$ as a natural number. The origin of these problems is that \lean{IsCyclotomicExtension} takes as first argument a set \lean{(S : Set ℕ+)} (this is a design decision made at the beginning of the formalization of the theory of cyclotomic extensions in \mathlib, and our work highlighted the need for it to be refactored, which has now been completed).

Instead of fixing a primitive $p$-th root of unity $\zeta$ once and for all, it is more convenient in \Lean{} to work with an unspecified primitive root.
\begin{lstlisting}
variable {ζ : K} (hζ : IsPrimitiveRoot ζ p)
\end{lstlisting}
One issue that appears is that now \lean{(ζ : K)} has type $K$, and not $\OK$ or $\OK^\times$, forcing us to use several coercions. To make the formalization process as smooth as possible, we usually consider \lean{(hζ.unit' : OOˣ)}, that is the same as $\zeta$, but seen as a unit of the ring of integers. After having set the relevant \lean{@[simp]} and \lean{@[norm\_cast]} tags, we can use various tactics like \lean{push\_cast} or \lean{field\_simp} to avoid mathematically irrelevant annoyances.

The ring of integers of $K$ is denoted by \lean{OO K}\href{https://github.com/leanprover-community/mathlib4/blob/4705244150fe57b4006d3393c5040d3d74e024b4/Mathlib/NumberTheory/NumberField/Basic.lean#L85-L95}{\faExternalLink}, and \mathlib{} contains an extensive library of results about it. For example, the fact that $\OK$ is a Dedekind domain\href{https://github.com/leanprover-community/mathlib4/blob/4705244150fe57b4006d3393c5040d3d74e024b4/Mathlib/NumberTheory/NumberField/Basic.lean#L276-L277}{\faExternalLink} and the fact that $\OK = \Z[\zeta]$\href{https://github.com/leanprover-community/mathlib4/blob/4705244150fe57b4006d3393c5040d3d74e024b4/Mathlib/NumberTheory/Cyclotomic/Rat.lean#L152-L159}{\faExternalLink} are both formalized: the former in the work~\cite{class_group}. To state that $p$ is a prime, we use \lean{[Fact p.Prime]} (here \lean{(p : ℕ)}): indeed \lean{Nat.Prime} is not a class, but using \lean{Fact} we can record primality as an instance and use typeclass inference. Mathematically, $p$ is regular if the ring of integers of any $p$-th cyclotomic extension has class number not divisible by $p$. This is clearly equivalent to the same property for \emph{one fixed model of a $p$-th cyclotomic extension of $\Q$} and we decided to formalize the definition this way, using \lean{CyclotomicField p ℚ}. This makes it simpler in practice to show that a given number is regular (the point of having a general $K$ in the project is that most of the theorems concern a cyclotomic extension, so we don't want them to hold only for \lean{CyclotomicField p ℚ}). Here is our definition of being a regular prime\href{https://github.com/leanprover-community/flt-regular/blob/34cb497affb3355f8db9bb06ae271cb4b5f10c3b/FltRegular/NumberTheory/RegularPrimes.lean#L25-L31}{\faExternalLink}.
\begin{lstlisting}
    /-- A natural number `n` is regular if `n` is coprime with the cardinal of the class group -/
    def IsRegularNumber (n : ℕ) [hn : Fact (0 < n)] : Prop :=
      n.Coprime <| Fintype.card <| ClassGroup (OO <| CyclotomicField ⟨n, hn.1⟩ ℚ)

    def IsRegularPrime (p : ℕ) [Fact p.Prime] : Prop := IsRegularNumber p
\end{lstlisting}
Strictly speaking the definition \lean{IsRegularPrime} is not needed and the main result can be stated using \lean{IsRegularNumber} only: we find nevertheless the use of \lean{IsRegularPrime} closer to informal mathematics. Note that we define the notion of being regular for any positive integer as being coprime with the class number of the corresponding cyclotomic extension, but this is used in practice only for prime numbers: in particular we do not claim that Kummer's lemma holds for regular numbers, it is for example false for $n=4$. (This is customary in formalized mathematics, since primality is not needed to state the condition, we do not assume it, even if this generality is not mathematically interesting.) Our main theorem
is then\href{https://github.com/leanprover-community/flt-regular/blob/34cb497affb3355f8db9bb06ae271cb4b5f10c3b/FltRegular/FltRegular.lean#L7-L17}{\faExternalLink} the following:
\begin{lstlisting}
    /-- Fermat's last theorem for regular primes. -/
theorem flt_regular {p : ℕ} [Fact p.Prime] (hreg : IsRegularPrime p) (hodd : p ≠ 2) :
    FermatLastTheoremFor p := by ...
\end{lstlisting}
Here \lean{FermatLastTheoremFor p}\href{https://github.com/leanprover-community/mathlib4/blob/4705244150fe57b4006d3393c5040d3d74e024b4/Mathlib/NumberTheory/FLT/Basic.lean#L51-L52}{\faExternalLink} is the statement, existing in \mathlib,
\begin{lstlisting}
∀ a b c : ℕ, a ≠ 0 → b ≠ 0 → c ≠ 0 → a ^ p + b ^ p ≠ c ^ p
\end{lstlisting}
Note that technically \lean{FermatLastTheoremFor} is a statement about natural numbers, but the equivalence to the same statement for integers or rationals already exists in \mathlib{}\href{https://github.com/leanprover-community/mathlib4/blob/4705244150fe57b4006d3393c5040d3d74e024b4/Mathlib/NumberTheory/FLT/Basic.lean#L84-L85}{\faExternalLink}.

\section{Case 1} \label{sec: case 1}
The proof of FLT for regular primes is split into two sub-cases. The so called ``first case'' of FLT is the following (we use $\Z$ in the project rather than $\N$ and we switch to natural numbers at the end to prove \lean{FermatLastTheoremFor})
\begin{theorem}[Case 1]\label{thm: case 1}
	Let $a,b,c \in \Z$ and $p$ a regular prime. If $p$ does not divide $a b c$,  then $a^p+b^p \ne c^p$.
\end{theorem}
Here is the formalized statement\href{https://github.com/leanprover-community/flt-regular/blob/34cb497affb3355f8db9bb06ae271cb4b5f10c3b/FltRegular/CaseI/Statement.lean#L250-L252}{\faExternalLink}
\begin{lstlisting}
theorem caseI {a b c : ℤ} {p : ℕ} [Fact p.Prime] (hreg : IsRegularPrime p)
    (caseI : ¬↑p ∣ a * b * c) : a ^ p + b ^ p ≠ c ^ p := ...
\end{lstlisting}
Note that for $p=2$ it is easy to see that case 1 cannot hold since the sum of two odd numbers is necessarily even.  The formalization of this result was completed on October 13th 2022, a little less than one year later than the starting date of the project. This work is described in~\cite{case1}, and we will only briefly list the main bits we formalized: 

\begin{enumerate}
    \item Definitions of cyclotomic fields and cyclotomic extensions.\footnote{To date these are the only explicit examples of number fields in \mathlib, although some work in the quadratic case has been done \cite{formclassgps}.}
    \item Explicit formulas for norms, trace (previously defined in \cite{class_group}) as well as defining discriminants for number fields.
    \item Computing the discriminant and ring of integers for $p$-power cyclotomic fields.
    \item Specific results about units in these cyclotomic fields. For example that for $p \ne 2$ a prime, every unit $u \in \Z[\zeta_p]^\times$ can be written as $u=x \zeta_p^n$ for some $n \in \Z$ and $x \in \Z[\zeta_p]^\times$ such that $x \in \mathbb{R}$.
    \item The proof of FLT for $n=3$, which was done independently by one of us (RVdV). Note this case had already been done previously in Isabelle \cite{Isabelle_flt34}.
\end{enumerate}

Our formalization was entirely in \Leant{}, and shortly after the end of this task, the transition from \Leant{} to \Leanf{} started. 
Once \mathlib{} moved completely to \Leanf{}, we started updating \fltregular. For this, we used \href{https://github.com/leanprover-community/mathport}{mathport}, a tool developed by the \mathlib{} community (especially by Mario Carneiro and Gabriel Ebner) to port \Leant{} projects to \Leanf{}. Porting \fltregular, we encountered mainly two issues.
\begin{itemize}
    \item Coercions (especially from $\OK^\times$ to $\OK$ to $K$) are omnipresent in our project, see~\cite[Section~3]{case1}. \Leanf{} handles coercions in a completely different way from \Leant{}, so part of the code had to be rewritten by hand. Although tedious, the same issue had appeared in porting \mathlib{}, so the solutions were well understood.
    \item The way typeclass inference handles multiple inheritance changed in \Leanf{}, see for example~\cite{eric}. This caused slowdown in \fltregular and more generally in \mathlib{}\footnote{In the beginning we simply increased the time limit allocations in the project, via the \lean{set_option maxHeartbeats} and \lean{set_option synthInstance.maxHeartbeats} commands, and we worked with the slow code, but this was only a temporary solution.}. Thanks also to our work (and of course of the \mathlib{} community at large), these issues are now solved in \mathlib{}. These solutions found in \fltregular{} have been now almost completely integrated into mathlib (see Section~\ref{sec: form stuff} below for more details).
\end{itemize}
Another important aspect of our work, described in~\cite[Section 4]{case1}, is integration into \mathlib{}. We now have more than 200 pull requests merged that were originally in \fltregular. Notably, almost all the prerequisites needed to formalize the proof of Case 1 are now in \mathlib{} (with the exception of homogenization of polynomials) and we have already started the same process for Case 2. Even if the whole proof of FLT for regular primes is unlikely to be included into \mathlib{} (being very technical and specific) we expect that all prerequisites will be upstreamed. Besides improving \mathlib{} itself, we saw, especially during the port, that this process reduces the work of maintaining the code. This process can be very tedious due to the continuous development/expansion of \mathlib{}, meaning that code is commonly refactored (as better methods/approaches are developed) which leads to projects, not yet included into mathlib, breaking when updated to newer versions of \mathlib{}. This is in contrast to other systems with more structured numbered releases.

\section{Case 2} \label{sec: case 2}
Given the formalization of Theorem~\ref{thm: case 1} outlined in Section~\ref{sec: case 1}, to formalize the proof of FLT for regular primes it is enough to consider the following (which is known as Case 2 of FLT for regular primes):

\begin{theorem}[Case 2]\label{thm: case 2}
	Let $p \neq 2$ be a regular prime. Then $x^p+y^p=z^p$ has no solutions such that
	\begin{itemize}
	    \item $xyz \neq 0$.
	    \item $\gcd(x, y, z) = 1$.
	    \item $p \mid xyz$.
	\end{itemize}
\end{theorem}
Our formalization is the following\href{https://github.com/leanprover-community/flt-regular/blob/34cb497affb3355f8db9bb06ae271cb4b5f10c3b/FltRegular/CaseII/Statement.lean#L82-L85}{\faExternalLink}
\begin{lstlisting}
/-- CaseII. -/
theorem caseII {a b c : ℤ} {p : ℕ} [hpri : Fact p.Prime]
    (hreg : IsRegularPrime p) (hodd : p ≠ 2)
    (hprod : a * b * c ≠ 0) (hgcd : ({a, b, c} : Finset ℤ).gcd id = 1)
    (caseII : ↑p ∣ a * b * c) : a ^ p + b ^ p ≠ c ^ p := by ...
\end{lstlisting}
Here is the strategy of the proof that we have formalized (where the main technical step follows~\cite[Section~V.7.1]{borevich}). Let $x$, $y$ and $z$ be coprime integers such that $p \mid xyz$ and $x^p+y^p=z^p$. First note that $p$ divides exactly one of $x,y,z$, and by an easy change of variables we can assume that $p \mid z$. Since we have to show that there is no nontrivial solution with $\gcd(x,y,z) = 1$ and $p \mid z$, it is enough to prove that there is no nontrivial solution such that $p \nmid y$ and $p \mid z$. The statement we are going to prove is then\href{https://github.com/leanprover-community/flt-regular/blob/34cb497affb3355f8db9bb06ae271cb4b5f10c3b/FltRegular/CaseII/Statement.lean#L50-L51}{\faExternalLink}
\begin{lstlisting}
lemma not_exists_Int_solution {p : ℕ} [Fact (p.Prime)] (hreg : IsRegularPrime p)
    (hodd : p ≠ 2) :
  ¬∃ (x y z : ℤ), ¬↑p ∣ y ∧ ↑p ∣ z ∧ z ≠ 0 ∧ x ^ p + y ^ p = z ^ p := by ...
\end{lstlisting}
In $\OK$ we have that $p=(1-\zeta)^{p-1} \varepsilon$ for some $\varepsilon \in \OK^\times$ and the norm of $\zeta - 1$ is $p$, in particular\href{https://github.com/leanprover-community/flt-regular/blob/34cb497affb3355f8db9bb06ae271cb4b5f10c3b/FltRegular/NumberTheory/Cyclotomic/MoreLemmas.lean#L78}{\faExternalLink} $p$ divides an integer $n$ in $\Z$ if and only if $\zeta - 1$ divides $n$ in $\OK$. It follows that it is enough to prove\href{https://github.com/leanprover-community/flt-regular/blob/34cb497affb3355f8db9bb06ae271cb4b5f10c3b/FltRegular/CaseII/Statement.lean#L26-L28}{\faExternalLink}
\begin{lstlisting}
lemma not_exists_solution' : {p : ℕ+} [Fact (p.Prime)]
  (hreg : (p : ℕ).Coprime <| Fintype.card <| ClassGroup (OO K))
  ¬∃ (x y z : OO K), ¬(hζ.unit' : OO K) - 1 ∣ y ∧ (hζ.unit' : OO K) - 1 ∣ z ∧ z ≠ 0 ∧
    x ^ (p : ℕ) + y ^ (p : ℕ) = z ^ (p : ℕ) := by ...
\end{lstlisting}
Note that here \lean{(p : ℕ+)} (since we need cyclotomic extensions) and then we are forced to consider the coercion \lean{(p : ℕ)} to take exponentiation. Also, \lean{hreg} is in terms of $K$ rather than \lean{CyclotomicField p ℚ} (mathematically the two notions are of course equivalent): this is because we are proving results about $K$, so it is more convenient to work with it directly. Since ultimately we will apply our results to the case \lean{K = CyclotomicField p ℚ}, this causes no problems.

Write $z = p^k z_0$ with $z_0 \in \Z$ such that $p \nmid z_0$. It suffices to prove that the more general equation
\[
x^p+y^p=\varepsilon(\zeta-1)^{p(m+1)}z_0^p
\]
has no nontrivial solutions with $x,y,z_0 \in \OK$ such that $y$ and $z_0$ are not divisible by $1-\zeta$. Here $m \ge 0$ and $\varepsilon \in \OK^\times$. The proof follows by a descent argument (which will crucially depend on Kummer's lemma) where one shows that if such a solution exists for some $m$ then one can construct a solution with $m$ replaced by $m-1$.

We now give more details about the proof, notably explaining where Kummer's lemma is needed. Assume we have $x,y,z \in \OK$ and $\varepsilon \in \OK^\times$ such that $y$ and $z$ are not divisible by $1-\zeta$ and that we have
\begin{equation} \label{eq: gen}
x^p+y^p=\varepsilon(1-\zeta)^{p(m+1)}z^p
\end{equation}
for some natural number $m$.  In what follows we will keep this choice of $\zeta$ fixed and when discussing an arbitrary $p$-th root of unity we will denote them by $\eta$.

Let $\gothz$ and $\gothp$ denote the ideals generated by $z$ and $\zeta - 1$ respectively. Note that $\gothp$ is the only prime ideal of $\OK$ above $p$. Then factoring~\eqref{eq: gen} as ideals we have
\[ \label{eq: gen fact}
\prod_{\eta \in \mu_p} (x+\eta y) = (\zeta-1)^{p(m + 1)} \gothz^p.
\]
It follows that each $x+\eta y$ is divisible by $(\zeta-1)$ in $\OK$\href{https://github.com/leanprover-community/flt-regular/blob/34cb497affb3355f8db9bb06ae271cb4b5f10c3b/FltRegular/CaseII/InductionStep.lean#L83-L84}{\faExternalLink}. The idea now is to divide each side by $(\zeta-1)^p$ and then after some manipulations, show that we can create a solution with a smaller power of $m$ on the right hand side. Following \cite{BorSha}, one then produces a solution to an equation of the form:

\begin{equation} \label{eq: formula almost}
u_1 x'^p + y'^p=u_2(\zeta-1)^{pm}z'^p.
\end{equation}
Since we already know that $m \geq 1$, we are done \emph{if we can prove that $u_1$ is a $p$-th power}. This is the crucial step where we need Kummer's lemma, Theorem~\ref{thm: kummer} below. Thanks to it, it is enough to show that $u_1$ is congruent to some integer modulo $p$. To prove this, note that since $(p) = \gothp^{p-1}$, there is\href{https://github.com/leanprover-community/flt-regular/blob/34cb497affb3355f8db9bb06ae271cb4b5f10c3b/FltRegular/NumberTheory/Cyclotomic/MoreLemmas.lean#L27}{\faExternalLink} a unit $u \in \OK^\times$ such that $p u = (\zeta - 1)^{p-1}$. Substituting in~\eqref{eq: formula almost}, we see that $p \mid u_1 x'^p + y'^p$. This implies that\href{https://github.com/leanprover-community/flt-regular/blob/34cb497affb3355f8db9bb06ae271cb4b5f10c3b/FltRegular/CaseII/InductionStep.lean#L376-L379}{\faExternalLink} there is some $a \in \OK$ such that $u_1$ is congruent to $a^p$ modulo $p$. We conclude since for any $x \in \OK$, we have that\href{https://github.com/leanprover-community/flt-regular/blob/34cb497affb3355f8db9bb06ae271cb4b5f10c3b/FltRegular/NumberTheory/Cyclotomic/MoreLemmas.lean#L58}{\faExternalLink} $x ^ p$ is congruent to an integer modulo $p$.

With the above setup and setting $\pi = \zeta-1$, the formalization of this argument concludes that\href{https://github.com/leanprover-community/flt-regular/blob/34cb497affb3355f8db9bb06ae271cb4b5f10c3b/FltRegular/CaseII/InductionStep.lean#L562-L564}{\faExternalLink} we have (see Section~\ref{sec: notation} for an explanation of the various notations):

\begin{lstlisting}
variable {K : Type*} {p : ℕ+} [hpri : Fact p.Prime] [Field K]
  [NumberField K] [IsCyclotomicExtension {p} ℚ K] (hp : p ≠ 2) [Fintype (ClassGroup (OO K))] 
  (hreg : (p : ℕ).Coprime <| Fintype.card <| ClassGroup (OO K)) {ζ : K} (hζ : IsPrimitiveRoot ζ p) 
  {x y z : OO K} {ε : (OO K)ˣ} {m : ℕ} (hy : ¬ hζ.unit'.1 - 1 ∣ y) (hz : ¬ hζ.unit'.1 - 1 ∣ z)
  (e : x ^ (p : ℕ) + y ^ (p : ℕ) =   -- i.e. assume we have a solution with m+1
                ε * ((hζ.unit'.1 - 1) ^ (m + 1) * z) ^ (p : ℕ))

lemma exists_solution' : ∃ (x' y' z' : OO K) (ε₃ : (OO K)ˣ), ¬ π ∣ y' ∧ ¬ π ∣ z' ∧
    x' ^ (p : ℕ) + y' ^ (p : ℕ) = ε₃ * (π ^ m * z') ^ (p : ℕ) --then we get a solution with m
\end{lstlisting}

\section{Kummer's lemma} \label{sec: kummer}
Kummer's lemma states the following:
\begin{theorem} \label{thm: kummer}
Let $p$ be an odd regular prime and let $\zeta$ be a primitive $p$-th root of unity. If $u \in \OK^\times$ is a unit such that $u \equiv n \bmod p$ for an integer $n$, then there exists $v \in \OK^\times$ such that $u = v^p$.
\end{theorem}
Nowadays, one can use more modern tools, such as class field theory to prove Kummer's lemma as follows.
\begin{proof}[Proof of Theorem~\ref{thm: kummer}]
Let us consider the extension $K(u^{1/p})/K$: it is a Kummer, hence abelian, extension. By a rather elementary argument (see~\cite[second proof of Theorem 5.36]{washington} or the proof at the end of this section), the extension is everywhere unramified (here is where the assumption on $u$ is needed), so $K(u^{1/p})$ is contained in $L$, the Hilbert class field of $K$. By class field theory, $[L : K]$ is the class number of $K$ and in particular $p \nmid [L : K]$. It follows that the degree of $K(u^{1/p})/K$ cannot be $p$: the only other possibility is that the degree is $1$, so $K(u^{1/p}) = K$ and $u$ has a $p$-th root in $K$ and hence in $\OK^\times$.
\end{proof}
The formalization of global class field theory is a long-term project, but we are currently quite far from it, making it impossible to formalize such a proof. There are other proofs, for example using the $p$-adic regulator, but again their formalization is unfeasible at the moment. There is also Kummer's original proof \cite{kummer1} which uses the Bernoulli characterization of regular primes (see \ref{thm: bernouilli}), making the proof more elementary at the cost of proving \ref{thm: bernouilli}, which itself relies on versions of the analytic class number formula for cyclotomic fields \cite{kummer2}.

For these reasons we instead followed a more elementary proof relying on some basic ramification results and Hilbert's Theorems 90, 92 and 94 (see~\cite{Hilbert} for the original formulation). Those are historically the starting point of modern class field theory, so in a sense our formalized proof is the one above, but written in much more down-to-earth terms.

The formalization of the statement is the following (notations as in Section~\ref{sec: notation})\href{https://github.com/leanprover-community/flt-regular/blob/34cb497affb3355f8db9bb06ae271cb4b5f10c3b/FltRegular/NumberTheory/KummersLemma/KummersLemma.lean#L37-L40}{\faExternalLink}
\begin{lstlisting}
theorem eq_pow_prime_of_unit_of_congruent (u : (OO K)ˣ)
    (hcong : ∃ n : ℤ, (p : OO K) ∣ (u - n : OO K)) : ∃ v, u = v ^ (p : ℕ) := by ...
\end{lstlisting}
Note we decided to spell out explicitly the condition that $u$ is congruent to $n$ modulo $p$, without actually using $\equiv$ in the formal statement. This is of course equivalent, and we felt that the API was better suited to work with division (we could also have worked with ideals, but again this is more cumbersome).

We now describe in more detail the proof we formalized. Let $u \in \OK^\times$ and $n \in \Z$ such that $u \equiv n \bmod p$. Suppose for a moment that if $w \in \OK^\times$ is a unit such that $(\zeta -1)^p \mid w-1$ then there exists some $v \in K$ such that $v^p=w$. Admitting this claim, we see that with our setup we have $u \equiv a \mod p$ for some integer $a$. Now this means that $u^{p-1} \equiv 1 \mod p$, which means that $u^{p-1} \equiv 1 \mod (\zeta-1)^{p-1}$ and hence $u^{p-1} \equiv 1 \mod (\zeta-1)^p$. Using the claim we then have a $v'$ (which must necessarily be a unit) such that $v'^p = u^{p-1}$. So choosing $v=u/v'$ we see that $v^p = u$ as required.

So it remains to prove the claim, which is formalized as follows\href{https://github.com/leanprover-community/flt-regular/blob/34cb497affb3355f8db9bb06ae271cb4b5f10c3b/FltRegular/NumberTheory/KummersLemma/KummersLemma.lean#L14-L15}{\faExternalLink}
\begin{lstlisting}
theorem exists_pow_eq_of_zeta_sub_one_pow_dvd_sub_one {u : (OO K)ˣ}
    (hcong : (hζ.unit' - 1 :OO K) ^ (p : ℕ) ∣ (u : OO K) - 1) : ∃ v : K, v ^ (p : ℕ) = u := by...
\end{lstlisting}
To see this, assume for contradiction that the claim is false, so $u$ is not a $p$-th power and let $L = K(u^{1/p})$. The extension $L/K$ is a Kummer extension that by our assumption has degree $p$. Note that there was no theory of Kummer extensions in \mathlib{} when the project began, and we developed it from the ground up. This new material is now fully integrated in the library and it is fairly complete: for example \mathlib{} knows that $K(u^{1/p})/K$ is Galois with cyclic Galois group.

We now claim that $L/K$ is unramified. Again we had to develop the theory of unramified extensions from scratch, and this material is now in the process of being integrated into \mathlib{}. The formal statement we use to prove unramifiedness is the following\href{https://github.com/leanprover-community/flt-regular/blob/34cb497affb3355f8db9bb06ae271cb4b5f10c3b/FltRegular/NumberTheory/KummersLemma/Field.lean#L303-L304}{\faExternalLink}
\begin{lstlisting}
lemma isUnramified (L) [Field L] [Algebra K L]
    [IsSplittingField K L (X ^ (p : ℕ) - C (u : K))] : IsUnramified (OO K) (OO L) := by ...
\end{lstlisting}
Here $L$ is the splitting field of $X ^ p - u$ over the cyclotomic field $K$. Consider the polynomial
\[
P_u = \frac{((\zeta - 1)X - 1)^p+u}{(\zeta - 1)^p} \in K[X]
\]
First of all we have that\href{https://github.com/leanprover-community/flt-regular/blob/34cb497affb3355f8db9bb06ae271cb4b5f10c3b/FltRegular/NumberTheory/KummersLemma/Field.lean#L16-L18}{\faExternalLink} $P_u \in \OK[X]$, it is monic\href{https://github.com/leanprover-community/flt-regular/blob/34cb497affb3355f8db9bb06ae271cb4b5f10c3b/FltRegular/NumberTheory/KummersLemma/Field.lean#L58}{\faExternalLink} and has degree $p$\href{https://github.com/leanprover-community/flt-regular/blob/34cb497affb3355f8db9bb06ae271cb4b5f10c3b/FltRegular/NumberTheory/KummersLemma/Field.lean#L66}{\faExternalLink}. Moreover, in $K[X]$ we have that\href{https://github.com/leanprover-community/flt-regular/blob/34cb497affb3355f8db9bb06ae271cb4b5f10c3b/FltRegular/NumberTheory/KummersLemma/Field.lean#L72-L73}{\faExternalLink}
\[
P_u = \left(X - (\zeta - 1)^{-1} \right)^p + \left(\frac{u}{\zeta - 1} \right)^p.
\]
Since we are assuming that $u$ is not a $p$-th power we have that\href{https://github.com/leanprover-community/flt-regular/blob/34cb497affb3355f8db9bb06ae271cb4b5f10c3b/FltRegular/NumberTheory/KummersLemma/Field.lean#L101-L102}{\faExternalLink} $P_u$ is irreducible over $K$. Moreover\href{https://github.com/leanprover-community/flt-regular/blob/34cb497affb3355f8db9bb06ae271cb4b5f10c3b/FltRegular/NumberTheory/KummersLemma/Field.lean#L139-L143}{\faExternalLink}, the roots of $P_u$ are the
\[
\alpha_m \colonequals \frac{1-\zeta^m u^{1/p}}{\zeta - 1},
\]
for $m = 0, \ldots, p-1$ and so $P_u$ is the minimal polynomial of any of the $\alpha_m$ over $K$ and hence\href{https://github.com/leanprover-community/flt-regular/blob/34cb497affb3355f8db9bb06ae271cb4b5f10c3b/FltRegular/NumberTheory/KummersLemma/Field.lean#L220-L222}{\faExternalLink}, by Gauss's lemma, over $\OK$. An explicit computation shows that $X - \alpha_{m_1}$ and $X - \alpha_{m_2}$ are coprime if $m_1 \neq m_2$, and thus\href{https://github.com/leanprover-community/flt-regular/blob/34cb497affb3355f8db9bb06ae271cb4b5f10c3b/FltRegular/NumberTheory/KummersLemma/Field.lean#L228-L230}{\faExternalLink} $P_u$ is a separable polynomial\footnote{meaning it is coprime to its derivative.} over $\OL$ and in particular\href{https://github.com/leanprover-community/flt-regular/blob/34cb497affb3355f8db9bb06ae271cb4b5f10c3b/FltRegular/NumberTheory/KummersLemma/Field.lean#L265-L265}{\faExternalLink} it is separable over $\OK/I$ for any maximal ideal $I \subseteq \OK$. Since clearly $L = K(\alpha_m)$ and $\alpha_m$ has separable polynomial modulo any maximal ideal of $\OK$, we have that\href{https://github.com/leanprover-community/flt-regular/blob/34cb497affb3355f8db9bb06ae271cb4b5f10c3b/FltRegular/NumberTheory/Unramified.lean#L193-L197}{\faExternalLink} the extension $L/K$ is everywhere unramified. Note that the proof of unramifiedness, although not very complex, is quite intricate, and it required us to develop the theory of unramified extensions. In doing so we formalized, among other things the following (for arbitrary extensions of number fields $L/K$):
\begin{enumerate}
    \item Definition of the relative different ideal $\mathfrak{d}_{L/K} \colonequals  ((\OL)^*)^{-1}$ where
    \[
    M^* = \{ \alpha \in L | \; \Tr_{L/K} (\alpha M) \in \mathcal{O}_K\}.
    \]
    This is now in \mathlib{}\href{https://github.com/leanprover-community/mathlib4/blob/4705244150fe57b4006d3393c5040d3d74e024b4/Mathlib/RingTheory/DedekindDomain/Different.lean#L392-L396}{\faExternalLink}
    \item Proving that if $L/K$ is an extension of number fields and $S$ denote the set of $\alpha \in \OL$ such that $L=K(\alpha)$, then\href{https://github.com/leanprover-community/mathlib4/blob/4705244150fe57b4006d3393c5040d3d74e024b4/Mathlib/RingTheory/DedekindDomain/Different.lean#L534-L536}{\faExternalLink} 
    \[
    \mathfrak{d}_{L/K} = \left (m_{\alpha}'(\alpha) : \alpha \in S \right )
    \]
    where $m_\alpha$ is denotes the minimal polynomial of $\alpha$.
    \item If $\mathfrak{p}_K, \mathfrak{p}_L$ prime ideals in $L,K$ respectively, with $\mathfrak{p}_L^e \parallel \mathfrak{p}_K$ for $e >0$, then\href{https://github.com/leanprover-community/mathlib4/blob/4705244150fe57b4006d3393c5040d3d74e024b4/Mathlib/RingTheory/DedekindDomain/Different.lean#L627-L630}{\faExternalLink}
    \[
    \mathfrak{p}_L^{e-1} \mid \mathfrak{d}_{L/K}.
    \]
    \item If $L = K(\alpha)$ and $m_\alpha$ is separable modulo a prime ideal $\mathfrak{p}_K$, then\href{https://github.com/leanprover-community/flt-regular/blob/34cb497affb3355f8db9bb06ae271cb4b5f10c3b/FltRegular/NumberTheory/Unramified.lean#L193-L197}{\faExternalLink} $\mathfrak{p}_K$ is unramified in $\OL$.
\end{enumerate}
The situation is now the following: we have a unit $u \in \OK^\times$ that is not a $p$-th power and moreover the extension $K(u^{1/p})/K$ is everywhere unramified. We need to find a contradiction. Note that we have yet to use regularity of $p$. It is unlikely that one can finish the proof without some sort of class field theory result that allows one to deduce information about the class group from the existence of an unramified abelian extension. The one we formalize is Hilbert's Theorem 94, whose second part, Theorem~\ref{thm: Hilbert 94 2} below, finishes the proof since by regularity $p$ cannot divide the class number of $K$.

\section{Hilbert's Theorem 94} \label{sec: Hilbert94}
Hilbert's Theorem 94 has two parts to it. Continuing our backwards reasoning, we start with the second one, that we used above to finish the proof of unramifiedness.
\begin{theorem} \label{thm: Hilbert 94 2} [Hilbert's Theorem 94, part 2]
Let $L/K$ be an unramified cyclic finite extension of number fields of odd prime degree, then $[L:K]$ divides the class number of $K$.
\end{theorem}
\begin{proof}
To prove the theorem, assume that $[L:K]$ does not divide the class number and note that if $I$ is any ideal of $\OK$ such that $I_L$, the extension to $\OL$ of $I$, is principal then $I$ is automatically principal. Indeed, we have $\Norm_{L/K}(I_L)=I^{[L:K]}$ and, $I_L$ being principal means that $I^{[L:K]}$ is principal\href{https://github.com/leanprover-community/flt-regular/blob/34cb497affb3355f8db9bb06ae271cb4b5f10c3b/FltRegular/NumberTheory/Hilbert94.lean#L77-L78}{\faExternalLink}. Since $[L:K]$ is coprime to the class number we have that $I$ is principal\href{https://github.com/leanprover-community/flt-regular/blob/34cb497affb3355f8db9bb06ae271cb4b5f10c3b/FltRegular/NumberTheory/RegularPrimes.lean#L89-L93}{\faExternalLink}. In particular, we see that it is enough to show that there is a non-principal ideal of $\OK$ that becomes principal in $\OL$. This is precisely the content of part 1 of Hilbert's Theorem 94, Theorem~\ref{thm: Hilbert 94 1} below.
\end{proof}

A formalization of the above argument is the following\href{https://github.com/leanprover-community/flt-regular/blob/34cb497affb3355f8db9bb06ae271cb4b5f10c3b/FltRegular/NumberTheory/Hilbert94.lean#L124-L137}{\faExternalLink}. We include also the proof, that is very short and follows our informal proof above. The declaration \lean{exists_not_isPrincipal_and_isPrincipal_map} is Theorem~\ref{thm: Hilbert 94 1} below.
\begin{lstlisting}
theorem dvd_card_classGroup_of_isUnramified_isCyclic {K L : Type*}
    [Field K] [Field L] [NumberField K] [NumberField L] [Algebra K L]
    [FiniteDimensional K L] [IsGalois K L] [IsUnramified (OO K) (OO L)] [IsCyclic (L ≃ₐ[K] L)]
    (hKL : Nat.Prime (finrank K L))
    (hKL' : finrank K L ≠ 2) :
    finrank K L ∣ Fintype.card (ClassGroup (OO K)) := by
  obtain ⟨I, hI, hI'⟩ := exists_not_isPrincipal_and_isPrincipal_map K L hKL hKL'
  have := Fact.mk hKL
  rw [hKL.dvd_iff_not_coprime]
  exact fun h ↦ hI (isPrincipal_of_isPrincipal_pow_of_coprime h
    (Ideal.isPrincipal_pow_finrank_of_isPrincipal_map hI'))
\end{lstlisting}
Note that \lean{IsCyclic (L ≃ₐ[K] L)} is automatic by the other assumptions, but we find it more convenient to explicitly add it. In the applications, it is found by typeclass inference, so this does not cause any issue.

We now move on to the first part of Hilbert's Theorem 94, which states that:
\begin{theorem}\label{thm: Hilbert 94 1}[Hilbert's Theorem 94, part 1]
Let $L/K$ be an unramified cyclic extension of number fields of odd prime degree, then there is a non-principal ideal of $\OK$ that becomes principal in $\OL$.
\end{theorem}
\begin{proof}
This is a consequence of Hilbert's Theorem 92, Theorem~\ref{thm: Hilbert 92} below. Let $\eta \in \OK^\times$ be a unit with $\Norm_{L/K}(\eta) = 1$ and such that for all $w \in \OK^\times$ we have $\eta \neq w/\sigma(w)$, where $\sigma$ is a fixed generator of $\Gal(L/K)$. By Hilbert's Theorem 90, Theorem~\ref{thm: Hilbert 90 concr} below, there is $\beta \in \OK$ such that $\beta \neq 0$ and $\beta=\sigma(\beta)\eta$. We let $I \colonequals \OK \cap (\beta)$, the ideal of $\OK$ given by the ideal of $\OL$ generated by $\beta$. We claim that $I$ satisfies the required properties. Suppose for a moment that $I \OL = (\beta)$. This obviously implies that the extension of $I$ to $\OL$ is principal. If $I$ were principal generated by $\gamma$, there would be a unit $w \in \OK^\times$ such that $\beta = \gamma w$. Substituting into $\beta=\sigma(\beta)\eta$, since $\sigma(w)=w$ by $w \in \OK$, we get $\eta=\gamma/\sigma(\gamma)$ that is absurd. In particular $I$ is not principal.

We now prove that $I \OL = (\beta)$. Since $\sigma^{-1}(\beta)=\sigma^{-1}(\eta)\beta$, we have that\href{https://github.com/leanprover-community/flt-regular/blob/34cb497affb3355f8db9bb06ae271cb4b5f10c3b/FltRegular/NumberTheory/Hilbert94.lean#L24-L26}{\faExternalLink} $\sigma^{-1}(I) \OL = (\beta)$. Since $L/K$ is unramified, this implies that\href{https://github.com/leanprover-community/flt-regular/blob/34cb497affb3355f8db9bb06ae271cb4b5f10c3b/FltRegular/NumberTheory/Unramified.lean#L58-L60}{\faExternalLink} $I \OL = (\beta)$ (this is a technical point where we use the unramifiedness assumption to factor $I \OL$ into product of prime ideals). 
\end{proof}
Our formalization of the above theorem is the following\href{https://github.com/leanprover-community/flt-regular/blob/34cb497affb3355f8db9bb06ae271cb4b5f10c3b/FltRegular/NumberTheory/Hilbert94.lean#L111-L119}{\faExternalLink}.
\begin{lstlisting}
theorem exists_not_isPrincipal_and_isPrincipal_map (K L : Type*)
    [Field K] [Field L] [NumberField K] [NumberField L] [Algebra K L]
    [FiniteDimensional K L] [IsGalois K L] [IsUnramified (OO K) (OO L)] [h : IsCyclic (L ≃ₐ[K] L)]
    (hKL : Nat.Prime (finrank K L)) (hKL' : finrank K L ≠ 2) :
    ∃ I : Ideal (OO K), ¬I.IsPrincipal ∧ (I.map (algebraMap (OO K) (OO L))).IsPrincipal := by ...
\end{lstlisting}
The formal proof is a faithful translation of the informal one above.

\section{Hilbert's Theorem 92} \label{sec: Hilbert92}
The main technical input in our formalization of Kummer's lemma is the following result:
\begin{theorem}[Hilbert's Theorem 92] \label{thm: Hilbert 92}
Let $K/k$ be a cyclic extension of number fields of odd prime degree and let $\sigma$ be a generator of the Galois group. Then there exists a unit $u \in \OK^\times$ such that $\Norm_{K/k}(u)=1$ and for all $\varepsilon \in \OK^\times$ we have $u \ne \varepsilon/\sigma(\varepsilon)$.
\end{theorem}
We formalize the above statement as\href{https://github.com/leanprover-community/flt-regular/blob/34cb497affb3355f8db9bb06ae271cb4b5f10c3b/FltRegular/NumberTheory/Hilbert92.lean#L845-L848}{\faExternalLink}:
\begin{lstlisting}
lemma Hilbert92
    [Algebra k K] [IsGalois k K] [FiniteDimensional k K](hKL : Nat.Prime (finrank k K)) 
    (hpodd : finrank k K ≠ 2) (σ : K ≃ₐ[k] K) (hσ : ∀ x, x ∈ Subgroup.zpowers σ) :
    ∃ u : (OO K)ˣ, Algebra.norm k (u : K) = 1 ∧ ∀ ε : (OO K)ˣ, (u : K) ≠ ε / (σ ε : K) :=
\end{lstlisting}
Note that by Hilbert's Theorem 90, Theorem~\ref{thm: Hilbert 90 concr} below, we know that any such $u$ must be of the form $w/\sigma(w)$ for $w \in \OK$. So the content is in showing that $w$ cannot always be chosen to be a unit. This translates to saying that a certain cohomology group does not vanish as opposed to the usual vanishing result associated with Hilbert's Theorem 90. As remarked in~\cite[Page 76]{swinnerton}, the cohomological point of view does not make the proof of this theorem easier, so we instead follow the classical proof described therein, which also has the advantage of being more amenable to formalization given the current state of \mathlib, but does rely on Hilbert 90, which we do prove by cohomological means.

The construction of a unit satisfying the requirements of Hilbert's Theorem 92 will be a consequence of the existence of a particular set of units.
\subsection{Fundamental system of units}

\begin{definition}
Let $\sigma$ be a generator of $\Gal(K/k)$ (with $K/k$ as in Theorem~\ref{thm: Hilbert 92}) and let $r$ be $\rk_{\Z}(\Ok^\times)$. Let $\mathcal{U}'_{K/k}= \OK^\times / \Ok^\times$ and let $\mathcal{U}'_{K/k,\tors}$ denote the torsion subgroup. Let
\[
\mathcal{U}_{K/k} = \mathcal{U}'_{K/k} / \mathcal{U}'_{K/k,\tors},
\]
which we note is naturally a $\Z[\Gal(K/k)]$-module. Then a fundamental system of $r+1$ units is a choice of $r+1$ elements $\{h_i\}$ of $\mathcal{U}_{K/k}$ such that
\[
\mathcal{U}_{K/k}/ \langle h_i\rangle_{\Z[\Gal(K/k)]}
\]
is finite with minimal index.

\end{definition}

\begin{remark}
Note that this definition of fundamental system of units differs from perhaps the more traditional one, where the final quotient is by the $\Z$-module generated by the $h_i$ (as opposed to the $\Z[\Gal(K/k)]$-module we use).
\end{remark}

We are going to show that a fundamental system of units exists. Instead of working directly with $\mathcal{U}_{K/k}$, we find it more convenient to formalize a more general situation and then restrict to the specific setting above. We begin with any additive commutative group $G$ playing the role of $\mathcal{U}_{K/k}$, which we also assume is a $A$-module, where $A=\Z[\Gal(K/k)]$. Note that the natural notation for $\mathcal{U}_{K/k}$ is multiplicative rather than additive, but to use the language of modules we are forced to work with additive groups. 

\begin{remark}
For convenience in the \Lean code, we take $A=\Z[\zeta]$ identified to $\Z[\Gal(K/k)]$ via $\zeta \mapsto \sigma$ (recall that $\sigma$ is a fixed generator of the Galois group).
\end{remark}

We start by defining the notion of \emph{system of $s$ units}\href{https://github.com/leanprover-community/flt-regular/blob/34cb497affb3355f8db9bb06ae271cb4b5f10c3b/FltRegular/NumberTheory/SystemOfUnits.lean#L25-L28}{\faExternalLink} where $s$ is a natural number. It is just a set of $s$ elements of $G$ that are linearly independent over $A$.
\begin{lstlisting}
structure systemOfUnits (G : Type*) [AddCommGroup G] (s : ℕ) where
  units : Fin s → G  --A choice of s elements of G
  linearIndependent : LinearIndependent A units
\end{lstlisting}
We prove\href{https://github.com/leanprover-community/flt-regular/blob/34cb497affb3355f8db9bb06ae271cb4b5f10c3b/FltRegular/NumberTheory/SystemOfUnits.lean#L109}{\faExternalLink} that if $G$ is free as $\Z$-module of rank $s(p-1)$, then a system of $s$ units always exists.

We now move on to the definition of \emph{fundamental systems of $r$ units}. We start by introducing\href{https://github.com/leanprover-community/flt-regular/blob/34cb497affb3355f8db9bb06ae271cb4b5f10c3b/FltRegular/NumberTheory/Hilbert92.lean#L28-L31}{\faExternalLink} the notion of \emph{maximal system} of $s$ units.
\begin{lstlisting}
abbrev systemOfUnits.IsMaximal {s : ℕ} {p : ℕ+} {G : Type*} [AddCommGroup G]
    [Module A G] (sys : systemOfUnits (G := G) p s) :=
    Fintype (G ⧸ Submodule.span A (Set.range sys.units))
\end{lstlisting}
By definition, being maximal means that the quotient by the submodule generated by the elements of the system of units is finite. We then prove\href{https://github.com/leanprover-community/flt-regular/blob/34cb497affb3355f8db9bb06ae271cb4b5f10c3b/FltRegular/NumberTheory/Hilbert92.lean#L34-L35}{\faExternalLink} that if there is a system of $s$ units then there is a maximal one if $\rk_{\Z}(G)=s(p-1)$.

A system of $s$ units is \emph{fundamental}\href{https://github.com/leanprover-community/flt-regular/blob/34cb497affb3355f8db9bb06ae271cb4b5f10c3b/FltRegular/NumberTheory/Hilbert92.lean#L46-L49}{\faExternalLink} if it is maximal and the submodule generated by the elements of the system has index smaller than those generated by any other maximal system.
\begin{lstlisting}
def systemOfUnits.IsFundamental [Module A G] (h : systemOfUnits p G s) :=
  ∃ _ : h.IsMaximal, ∀ (S : systemOfUnits p G s) (_ : S.IsMaximal), h.index ≤ S.index
\end{lstlisting}
We then prove\href{https://github.com/leanprover-community/flt-regular/blob/34cb497affb3355f8db9bb06ae271cb4b5f10c3b/FltRegular/NumberTheory/Hilbert92.lean#L107-L109}{\faExternalLink} that a fundamental system of $s$ units always exists if $G$ is free as $\Z$-module of rank $s(p-1)$ . The key property of fundamental systems of units is the following
\begin{lemma}\label{key_lem}
Let $a_0,\dots,a_{s-1}$ be integers not all divisible by $p$ and $\{h_i\}$ a fundamental system of $s$ units, then
\[
\prod_{i=0}^{s-1} h_i^{a_i} \neq u/\sigma(u)
\]
for any $u \in \mathcal{U}_{K/k}$.
\end{lemma}
Our formalization is (note the additive notation)\href{https://github.com/leanprover-community/flt-regular/blob/34cb497affb3355f8db9bb06ae271cb4b5f10c3b/FltRegular/NumberTheory/Hilbert92.lean#L182-L184}{\faExternalLink} 
\begin{lstlisting}
lemma corollary [Module A G] (S : systemOfUnits p G s) (hs : S.IsFundamental)
    (a : Fin s → ℤ) (ha : ∃ i , ¬ (p : ℤ) ∣ a i) :
    ∀ g : G, (1 - zeta p) • g ≠ ∑ i, a i • S.units i := by ...
\end{lstlisting}
Here \lean{zeta : CyclotomicIntegers p := AdjoinRoot.root _} is the fixed root of the $p$-th cyclotomic polynomial in $A$, that corresponds to the primitive $p$-th root of unity $\zeta$ (and that acts as the fixed generator $\sigma$).

\subsection{Hilbert 91}
We now go back to our specific situation of a cyclic field extension $K/k$. We want to show that a system of fundamental $r+1$ units exists in the case $r = \rk_{\Z}(\Ok^\times)$ and $G = \mathcal{U}_{K/k}$. As $\mathcal{U}_{K/k}$ is torsion-free and hence free (as a $\Z$-module), this amounts to showing its rank is $(r+1)(p-1)$. This is essentially the content of Hilbert's Theorem 91. We begin by defining the group $\mathcal{U}'_{K/k}= \OK^\times / \Ok^\times$ as\href{https://github.com/leanprover-community/flt-regular/blob/34cb497affb3355f8db9bb06ae271cb4b5f10c3b/FltRegular/NumberTheory/Hilbert92.lean#L210-L211}{\faExternalLink}:
\begin{lstlisting}
def RelativeUnits (k K : Type*) [Field k] [Field K] [Algebra k K] :=
  ((OO K)ˣ ⧸ (MonoidHom.range <| 
                Units.map (algebraMap (OO k) (OO K) : (OO k) →* (OO K))))
\end{lstlisting}
Here \lean{MonoidHom.range <| Units.map (algebraMap (OO k) (OO K) : (OO k) →* (OO K)))}
denotes the image of the units in $k$ under the natural embedding into $\OK^\times$. Since we are going to work with a fixed generator of the Galois group, we find it more convenient to package together all the data we have\href{https://github.com/leanprover-community/flt-regular/blob/34cb497affb3355f8db9bb06ae271cb4b5f10c3b/FltRegular/NumberTheory/Hilbert92.lean#L331-L333}{\faExternalLink}.
\begin{lstlisting}
def relativeUnitsWithGenerator (_hp : Nat.Prime p) (_hKL : finrank k K = p) 
(σ : K ≃ₐ[k] K) (_hσ : ∀ x, x ∈ Subgroup.zpowers σ) :=
  RelativeUnits k K
\end{lstlisting}
This is the same as \lean{RelativeUnits}, but it contains the choice of a generator $\sigma$. Finally, we define our additive torsion-free group $\mathcal{U}_{K/k} = \mathcal{U}'_{K/k} / \mathcal{U}'_{K/k,\tors}$ as\href{https://github.com/leanprover-community/flt-regular/blob/34cb497affb3355f8db9bb06ae271cb4b5f10c3b/FltRegular/NumberTheory/Hilbert92.lean#L340-L342}{\faExternalLink}:
\begin{lstlisting}
local notation "G" =>
  Additive (relativeUnitsWithGenerator p hp hKL σ hσ) ⧸
    AddCommGroup.torsion (Additive (relativeUnitsWithGenerator p hp hKL σ hσ))
\end{lstlisting}
Here, if \lean{H} is a multiplicative group, \lean{Additive H} is the same group, but with additive notation. As explained above we are forced to use it, and it causes a little bit of friction, but all the relevant results to go from \lean{H} to \lean{Additive H} were already in \mathlib{}. We show\href{https://github.com/leanprover-community/flt-regular/blob/34cb497affb3355f8db9bb06ae271cb4b5f10c3b/FltRegular/NumberTheory/Hilbert92.lean#L294-L297}{\faExternalLink} that if we consider $G$ as a $\Z[X]$-module via $X \mapsto \sigma$, then it is torsion with respect to the subgroup generated by the cyclotomic polynomial. This endows\href{https://github.com/leanprover-community/flt-regular/blob/34cb497affb3355f8db9bb06ae271cb4b5f10c3b/FltRegular/NumberTheory/Hilbert92.lean#L570-L573}{\faExternalLink} $G$ with a structure of an $A$-module. Moreover, we show\href{https://github.com/leanprover-community/flt-regular/blob/34cb497affb3355f8db9bb06ae271cb4b5f10c3b/FltRegular/NumberTheory/Hilbert92.lean#L665}{\faExternalLink} that it has rank (as a $\Z$-module)  equal to $(r+1)(p-1)$ (this step is rather delicate: we need first of all to take care of the quotient by the torsion submodule and then we use that our extension is unramified at infinite places since $p$ is assumed to be odd).

Putting it all together, we can use our general existence result above to show that there is in fact a fundamental system of $r + 1$ units for $G$ as claimed by Hilbert's Theorem 91\href{https://github.com/leanprover-community/flt-regular/blob/34cb497affb3355f8db9bb06ae271cb4b5f10c3b/FltRegular/NumberTheory/Hilbert92.lean#L684-L685}{\faExternalLink}:
\begin{lstlisting}
lemma Hilbert91 :
  ∃ S : systemOfUnits p G (NumberField.Units.rank k + 1), S.IsFundamental :=
  systemOfUnits.IsFundamental.existence p hp G 
  (NumberField.Units.rank k + 1) (finrank_G p hp hKL σ hσ)
\end{lstlisting}

\subsection{Proving Hilbert's Theorem 92}
We can now finish the proof of Hilbert's Theorem 92, reducing the proof of Fermat's Last Theorem for regular primes to Hilbert's Theorem 90.

Let $K/k$ be as in~\ref{thm: Hilbert 92}. Remember that in particular $p$ is an odd prime and that $\sigma$ is a generator of $\Gal(K/k)$. Recall that $r = \rk_{\Z}(\Ok^\times)$ and let $h_i$ be a fundamental system of $r+1$ units for $\mathcal{U}_{K/k}$, that exists by Hilbert's Theorem 91 above. We denote by $H_i$ fixed lifts to elements of $\OK^\times$. We fix $h \in \N$ such that $k$ contains a $p^h$-th root of unity $\nu$ but no $p^{h+1}$-th root (note we allow $h$ to be zero). Let $\xi = \nu^{p^{h-1}}$ which is now a $p$-th root of unity as above, with the convention that if $h=0$ we have $\xi=\nu=1$.\footnote{In the formalization this convention is automatic, as $0-1=0$ for natural numbers in \Lean.}

With these notations the proof proceeds as follows (in the formalization most of the statements are actually more general than those below, for example stated for any family of $r$ elements of $\OK^\times$ when the fact that being a fundamental system of $r$ units is not needed):
\begin{enumerate}
    \item By Hilbert's Theorem 90, Theorem~\ref{thm: Hilbert 90 concr} below, we know that, since $\xi$ has norm $1$, we can write $\xi = \varepsilon /\sigma( \varepsilon )$ for some $ \varepsilon \in \OK$.
    \item  We can assume that there is some $E \in \OK^\times$ such that
    \[
    \xi = E / \sigma(E)
    \]
    as otherwise we could take $\xi$ as the element required by~\ref{thm: Hilbert 92}.
    \item Note that from $ \xi = E / \sigma(E)$ we have that\href{https://github.com/leanprover-community/flt-regular/blob/34cb497affb3355f8db9bb06ae271cb4b5f10c3b/FltRegular/NumberTheory/Hilbert92.lean#L626-L629}{\faExternalLink} $\Norm_{K/k}(E)=E^p$.
    \item Let $\eta_i = \Norm_{K/k}(H_i)$ for $i \in \{0, \dots, r\}$ and let $\eta_{r+1} := \Norm_{K/k}(E)=E^p$ (again most of the formal statements are for general $\eta_i$).
    \item \label{eq1} There exist\href{https://github.com/leanprover-community/flt-regular/blob/34cb497affb3355f8db9bb06ae271cb4b5f10c3b/FltRegular/NumberTheory/Hilbert92.lean#L512-L517}{\faExternalLink} $a, a_i \in \Z$ for $i \in \{0,\dots,r+1\}$ such that
    \[
    \prod_{i=0}^{r+1} \eta_i^{a_i} = \nu^{ap}
    \]
    with $a_{i_0}$ not divisible by $p$ for some $i_0$. Moreover, if $\nu=1$ (i.e. $h=0$) then we can take $i_0$ to be different from $r+1$.
    \item With these exponents $a_i$ in hand we can now construct our element. Set\href{https://github.com/leanprover-community/flt-regular/blob/34cb497affb3355f8db9bb06ae271cb4b5f10c3b/FltRegular/NumberTheory/Hilbert92.lean#L543-L544}{\faExternalLink}
    \[
    J \colonequals \nu^{-a}\prod_{i=0}^{r+1} H_i^{a_i}
    \]
    where we set $H_{r+1} = E$. We need to show it satisfies the required properties.
    \item Taking its norm we have that\href{https://github.com/leanprover-community/flt-regular/blob/34cb497affb3355f8db9bb06ae271cb4b5f10c3b/FltRegular/NumberTheory/Hilbert92.lean#L539-L545}{\faExternalLink}
    \[
    \Norm_{K/k}(J) = \left( \prod_{i=0}^{r+1} \eta_i^{a_i} \right )\nu^{-ap}
    \]
    which is $1$ by the condition on the $a_i$'s in~\eqref{eq1}.
    \item Now if $J$ was of the form $\varepsilon/\sigma(\varepsilon)$ for some unit $\varepsilon$, then the same would be true of its image in $\mathcal{U}_{K/k}$. Its image is 
    \[
    \nu^{-a}\prod_{i=0}^r h_i^{a_i} \cdot E^{a_{r+1}}.
    \]
    Now, since $\nu \in \Ok^\times$, its image in the quotient is $1$ (remember that here we are using multiplicative notation). Moreover, since $E^p = \Norm_{K/k}(E)$ also belongs to $\Ok^\times$ we have that its image is $1$, so the image of $E$ is torsion. Being $\mathcal{U}_{K/k}$ torsion-free, we have that $E$ becomes equal to $1$ in the quotient and the image of $J$ is
    \[
    \prod_{i=0}^r h_i^{a_i}.
    \]
    \item By Lemma~\ref{key_lem} to conclude the proof it is enough to show that not all of the $a_i$ for $i \in \{0,\dots,r\}$ are divisible by $p$. Assume for contradiction this is the case, then $a_{r+1}$ must be divisible by $p$. Now, this means $h \ne 0$ since otherwise this would contradict our choice of $a_i$'s. But now this means that $\eta_{r+1}=\Norm_{K/k}(E)=E^p$ is the $p$-th power of a unit in $k$ (this follows from~\eqref{eq1}). In particular $E$ is a \emph{unit} now in $k$. But this means that $\sigma(E)=E$, in particular $\xi=1$, but this cannot be the case as we have already shown $h \ne 0$. 
\end{enumerate}
The proof of the last two points are inlined in the proof of\href{https://github.com/leanprover-community/flt-regular/blob/34cb497affb3355f8db9bb06ae271cb4b5f10c3b/FltRegular/NumberTheory/Hilbert92.lean#L726-L727}{\faExternalLink} \lean{almostHilbert92}, that is Hilbert's Theorem 92 with the additional assumption that $K/k$ is unramified at all infinite places.
\begin{remark}
Note that we have slightly simplified the proof of Hilbert's Theorem 92 as found in~\cite{Hilbert}, which splits into two cases, the first when $K$ does not contain a $p$-th root and the second where it contains a $p^h$-th root (but not a $p^{h+1}$-th root) for some $h$. Our proof instead allows $h=0$ and so follows the second case. 
\end{remark}

\section{Hilbert's Theorem 90} \label{sec: Hilbert90}
Hilbert's Theorem 90 is a now a classical result in Galois cohomology, and \mathlib{} already contains a version of (a generalization of) it, thanks to the work of Amelia Livingston, see~\cite{amelia}. We explain in this section the precise statement we need, and how to get there given what is already in \mathlib{}.

The statement we needed above is the following, matching Hilbert's original formulation.
\begin{theorem}[Hilbert's Theorem 90, concrete version] \label{thm: Hilbert 90 concr}
Let $L/K$ be a cyclic extension of fields and let $\sigma \in \Gal(L/K)$ be a generator. If $\eta \in L$ is such that $\Norm_{L/K}(\eta)=1$, then there exists an $\varepsilon \in \OL$ such that $\varepsilon \neq 0$ and $\eta \sigma(\varepsilon) = \varepsilon$.
\end{theorem}
\begin{proof}
By clearing denominators, it is enough to prove that there is $\varepsilon \in L$ such that $\varepsilon \neq 0$ and $\eta = \varepsilon/\sigma(\varepsilon)$. We then deduce the theorem from Noether's generalization of Hilbert's Theorem 90, Theorem~\ref{thm: Hilbert 90 coho} below. Let $c \colon \Gal(L/K) \to L^\times$ be the function\href{https://github.com/leanprover-community/flt-regular/blob/34cb497affb3355f8db9bb06ae271cb4b5f10c3b/FltRegular/NumberTheory/Hilbert90.lean#L18}{\faExternalLink} $\tau \mapsto \prod_{i = 0}^n \sigma^i(\eta)$, where $n \in \N$ is the unique natural number smaller than the order of $\Gal(L/K)$ such that $\sigma^n = \tau$. Using that $\Norm_{L/K}(\eta)= \prod_{\tau \in \Gal(L/K)} \tau(\eta) =1$ we have that\href{https://github.com/leanprover-community/flt-regular/blob/34cb497affb3355f8db9bb06ae271cb4b5f10c3b/FltRegular/NumberTheory/Hilbert90.lean#L99}{\faExternalLink} $c$ is a cocycle, i.e. $c(\tau_1 \tau_2) = \tau_1(c(\tau_2))c(\tau_1)$. Since $H^1(G, L^\times)$ is trivial by Theorem~\ref{thm: Hilbert 90 coho}, we have that $c$ is a coboundary. By definition this means that there is $x \in L$ such that $c(\tau) = \tau(x)/x$ for all $\tau \in \Gal(L/K)$ and in particular $c(\sigma) = \sigma(x)/x$. Since $c(\sigma) = \eta$, we have that $x\eta = \sigma(x)$ and $\varepsilon = x^{-1}$ satisfies the conditions of the theorem.
\end{proof}
Here is our formalization\href{https://github.com/leanprover-community/flt-regular/blob/34cb497affb3355f8db9bb06ae271cb4b5f10c3b/FltRegular/NumberTheory/Hilbert90.lean#L132-L133}{\faExternalLink}
\begin{lstlisting}
variable {K L : Type*} [Field K] [Field L] [Algebra K L]
  [IsGalois K L] [FiniteDimensional K L]
  {A B : Type*} [CommRing A] [CommRing B] [Algebra A B] [Algebra A L] [Algebra A K]
  [Algebra B L] [IsScalarTower A B L] [IsScalarTower A K L] [IsFractionRing A K]
  [IsDomain A] [IsIntegralClosure B A L] [IsDomain B]

lemma Hilbert90_integral (σ : L ≃ₐ[K] L) (hσ : ∀ x, x ∈ Subgroup.zpowers σ)
    (η : B) (hη : Algebra.norm K (algebraMap B L η) = 1) :
    ∃ ε : B, ε ≠ 0 ∧ η * galRestrict A K L B σ ε = ε := by ...
\end{lstlisting}
Note that in the formalization, instead of working with $\OK$ and $\OL$ (the integral closure of $\Z$ in $K$ and $L$ respectively), we decided to work with general integral domains $A$ and $B$ such that $B$ is in the integral closure of $A$ in $L$; see Section~\ref{sec: form stuff} for more details about this design decision.

The formal proof follows the informal one, clearing denominators (a step that in \Lean{} is not completely trivial since we need to consider the restriction of $\sigma$ to $B$) and deducing the theorem from\href{https://github.com/leanprover-community/flt-regular/blob/34cb497affb3355f8db9bb06ae271cb4b5f10c3b/FltRegular/NumberTheory/Hilbert90.lean#L104}{\faExternalLink}
\begin{lstlisting}
lemma Hilbert90 {η : L} (hη : Algebra.norm K η = 1) :
    ∃ ε : L, η = ε / σ ε := by ...
\end{lstlisting}
Note that even if technically the statement does not contain the fact that $\varepsilon \neq 0$, this is automatic: if $\varepsilon = 0$, then $\sigma(\varepsilon) = 0$ and hence $\varepsilon /\sigma(\varepsilon) = 0$ (remember that in \Lean{} division by $0$ returns $0$), so $\eta = 0$, which is impossible because $\Norm_{L/K} (\eta) = 1$.

Everything is now reduced to the following
\begin{theorem}[Noether's generalization of Hilbert's Theorem 90] \label{thm: Hilbert 90 coho}
Let $L/K$ be a finite Galois extension of fields with Galois group $G$. Then $H^1(G, L^\times)$ is trivial.
\end{theorem}
This is a nowadays classical result in group cohomology, and we will not recall the proof. Thanks to the formalization project~\cite{amelia} by Amelia Livingston, it is already present in \mathlib{}\href{https://github.com/leanprover-community/mathlib4/blob/4705244150fe57b4006d3393c5040d3d74e024b4/Mathlib/RepresentationTheory/GroupCohomology/Hilbert90.lean#L76-L81}{\faExternalLink}
\begin{lstlisting}
/-- Noether's generalization of Hilbert's Theorem 90: given a finite extension of fields and a function `f : Aut_K(L) → Lˣ` satisfying `f(gh) = g(f(h)) * f(g)` for all `g, h : Aut_K(L)`, there exists `β : Lˣ` such that `g(β)/β = f(g)` for all `g : Aut_K(L).` -/
theorem isMulOneCoboundary_of_isMulOneCocycle_of_aut_to_units
    (f : (L ≃ₐ[K] L) → Lˣ) (hf : IsMulOneCocycle f) :
    IsMulOneCoboundary f := by ...
\end{lstlisting}
here \lean{IsMulOneCocycle} and \lean{IsMulOneCoboundary} are defined as follows\href{https://github.com/leanprover-community/mathlib4/blob/4705244150fe57b4006d3393c5040d3d74e024b4/Mathlib/RepresentationTheory/GroupCohomology/LowDegree.lean#L606-L608}{\faExternalLink}\href{https://github.com/leanprover-community/mathlib4/blob/4705244150fe57b4006d3393c5040d3d74e024b4/Mathlib/RepresentationTheory/GroupCohomology/LowDegree.lean#L658-L660}{\faExternalLink} (where $M$ is any $G$ module, for a group $G$)
\begin{lstlisting}
def IsMulOneCocycle (f : G → M) : Prop := ∀ g h : G, f (g * h) = g • f h * f g

def IsMulOneCoboundary (f : G → M) : Prop := ∃ x : M, ∀ g : G, g • x / x = f g
\end{lstlisting}
Note that group cohomology here is defined in very concrete terms, as cocycles modulo coboundaries, and not via abstract machinery like derived functors (that exist in \mathlib{}). Even if at some point \mathlib{} will need the connection between these two definitions, the current design choice is very convenient for us, since it allowed us to prove the explicit theorem we need quickly.

\section{Implementation issues} \label{sec: form stuff}
We now describe two implementation issues we encountered during our formalization project. As the problem of the coercions from $\OK^\times$ to $\OK$ to $K$ and the diamonds related to characteristic zero fields are already discussed in~\cite[Section~3]{case1}, we will not discuss those.

After the port to \Leanf{}, we noticed that the project was very slow in several places, making it almost impossible to progress with the formalization. In collaboration with the \mathlib{} community (especially Matthew Ballard, Kevin Buzzard and Floris van Doorn), we identified that the main bottleneck was the definition of the ring of integers of a number field in \mathlib{}, that was
\begin{lstlisting}
  def RingOfIntegers := integralClosure ℤ K
  @[inherit_doc] scoped notation "OO" => NumberField.ringOfIntegers
\end{lstlisting}
In general, if $A$ is an $R$-algebra, the type of \lean{integralClosure R A} is \lean{Subalgebra R A}. This has the advantage that a lot of instances, for example \lean{CommRing (OO K)}, are found automatically by \Lean{}. On the other hand, it causes the following drawback. \mathlib{} contains the following instance\href{https://github.com/leanprover-community/mathlib4/blob/4705244150fe57b4006d3393c5040d3d74e024b4/Mathlib/Algebra/Algebra/Subalgebra/Basic.lean#L758-L760}{\faExternalLink}, where $\alpha$ is a type endowed with \lean{SMul A α}, that is a scalar multiplication by $A$.
\begin{lstlisting}
  instance [SMul A α] (S : Subalgebra R A) : SMul S α := ...
\end{lstlisting}
Mathematically this simply means that if we know how to multiply elements of $\alpha$ by any \lean{(a : A)} we can also multiply elements of $\alpha$ by any \lean{(s : S)}. In particular, any time an instance like \lean{SMul (OO K) (OO K)} is needed (for example looking for \lean{Module (OO K) (OO K)}), \Lean{} will look for an instance of \lean{SMul K (OO K)} (a mathematically meaningless problem), and this search will of course fail. What we realized is that this search is rather slow to fail, and, due to the intricacies of the algebra hierarchy in \mathlib{}, it is performed many times. We tried to manually lower the priority of the above instance, but this did not solve the problem. Indeed it is sometimes difficult to control the path chosen by the typeclass inference system, and this solution would not scale anyway. In the pull request \href{https://github.com/leanprover-community/mathlib4/pull/12386/}{\#12386} we then decided to change the definition of \lean{ringOfIntegers} to\href{https://github.com/leanprover-community/mathlib4/blob/4705244150fe57b4006d3393c5040d3d74e024b4/Mathlib/NumberTheory/NumberField/Basic.lean#L85-L95}{\faExternalLink}
\begin{lstlisting}
  def RingOfIntegers : Type _ := integralClosure ℤ K
\end{lstlisting}
The only difference with the above one is that the type of \lean{OO K} is now \lean{Type _} rather than \lean{Subalgebra ℤ K}. To be precise, \lean{OO K} is now the underlying type of \lean{integralClosure ℤ K}: this means that, for example, the following instance\href{https://github.com/leanprover-community/mathlib4/blob/4705244150fe57b4006d3393c5040d3d74e024b4/Mathlib/NumberTheory/NumberField/Basic.lean#L101-L102}{\faExternalLink}
\begin{lstlisting}
  instance : CommRing (OO K) := inferInstanceAs (CommRing (integralClosure _ _))
\end{lstlisting}
has to be added by hand (here the command \lean{inferInstanceAs} allows \Lean{} to see through the definition of \lean{RingOfIntegers}). On the other hand, it also means that the instance from \lean{SMul K (OO K)} to \lean{SMul (OO K) (OO K)} is not even considered when solving \lean{Module (OO K) (OO K)}. This simple modification made essentially all \mathlib's files related to ring of integers of number fields much faster (see the \href{http://speed.lean-fro.org/mathlib4/compare/c4507a75-3c32-47fe-99bb-ce918c649d8b/to/12429d19-2f18-4b6f-b762-2f3668b995b7}{benchmark results}) and drastically improved the speed of \fltregular. Note that the number of instances that have to be added by hand is quite limited, and the gain in speed is huge: we expect that a similar modification will be needed in \mathlib{} in the future for other analogous situations, for example in the definition of the adele ring of a global field.

Another issue is related to field extensions: if $L/K$ is such an extension, then $\OL$ is naturally an $\OK$-algebra. If moreover $K$ and $L$ are number fields, it is easy to prove that mathematically $\OL$ ``is'' the integral closure of $\OK$ in $L$. However, it is not possible in \Lean{} to prove the naive equality of types \lean{OO L = integralClosure (OO K) L}, since the correct formal translation of this result is a datum of an isomorphism (as $\OK$-algebras) between $\OL$ and the integral closure of $\OK$ in $L$. Working with such an isomorphism can be annoying, as it will appear several times. To avoid this issue, we decided to work more generally in the so-called \emph{$AKLB$ setting}:
\begin{lstlisting}
  variable {A K L B : Type*} [CommRing A] [CommRing B] [Field K] [Field L]
    [Algebra A B] [Algebra A L] [Algebra A K] [Algebra B L]
    [IsScalarTower A B L] [IsScalarTower A K L] [IsFractionRing A K] [IsIntegralClosure B A L]
\end{lstlisting}
Any result holding in this setting will apply both to \lean{B = OO L} and to \lean{B = integralClosure (OO K) L}, allowing very often to avoid the use of the isomorphism above (one may assume \lean{IsIntegralClosure ℤ A K} to state that $A$ ``is'' $\OK$). This setting is very common in algebraic number theory, and we believe this way of formalizing it should be used in \mathlib{} every time it is possible.

\section{Future work} \label{sec: future}
As explained in~\cite[Section~4]{case1}, one peculiar aspect of our project is that the integration into \mathlib{} is happening in real time (for example this was not happening at all for the LTE). Even if we don't think the full proof of Fermat's Last Theorem for regular primes should be included into \mathlib{} (as it is very technical and mathematically superseded by more modern approaches), most of the prerequisites we formalized are fundamental results in algebraic number theory and should move to the main library. We estimate that of the approximately $10k$ lines of the project, around $6k$ lines should end up in mathlib of which roughly $3k$ are already present. The process of opening pull requests from \fltregular to \mathlib{} never stopped, and we expect that all the relevant material will end up in \mathlib{} at some point. One notable example of this process is the proof of Fermat's Last Theorem for $n = 3$, that is now in \mathlib{}\href{https://github.com/leanprover-community/mathlib4/blob/4705244150fe57b4006d3393c5040d3d74e024b4/Mathlib/NumberTheory/FLT/Three.lean#L741-L743}{\faExternalLink}. This rather nontrivial case, along with $n = 4$ (also in \mathlib{}\href{https://github.com/leanprover-community/mathlib4/blob/4705244150fe57b4006d3393c5040d3d74e024b4/Mathlib/NumberTheory/FLT/Four.lean#L294-L298}{\faExternalLink}), will be required for Kevin Buzzard's \href{https://github.com/ImperialCollegeLondon/FLT}{project} of formalizing a proof of the full Fermat's Last Theorem, as modern proofs only cover prime exponents $p \ge 5$.

Concerning future formalizations, one natural question that is left by the current status of \fltregular is that the condition for a prime of being regular is rather abstract, and difficult to prove in practice. Indeed, the computation of the class number of $\Q(\e^{\frac{2\pi i}{n}})$ is a difficult problem, even on paper. Thanks to Xavier Roblot's work, \cite{acnf}, Minkowski's bound is in \mathlib{}\href{https://github.com/leanprover-community/mathlib4/blob/d376bfc5782e90bafe2815b936ece41bf3b69f05/Mathlib/NumberTheory/NumberField/ClassNumber.lean#L48-L51}{\faExternalLink}. In particular it is easy to prove that both $\Z[\e^{\frac{2\pi i}{3}}]$\href{https://github.com/leanprover-community/mathlib4/blob/4705244150fe57b4006d3393c5040d3d74e024b4/Mathlib/NumberTheory/Cyclotomic/PID.lean#L29-L30}{\faExternalLink} and $\Z[\e^{\frac{2\pi i}{5}}]$\href{https://github.com/leanprover-community/mathlib4/blob/4705244150fe57b4006d3393c5040d3d74e024b4/Mathlib/NumberTheory/Cyclotomic/PID.lean#L43-L44}{\faExternalLink} are principal ideal domains (since the Minkowski bound is $1$ in these cases), and hence $3$\href{https://github.com/leanprover-community/flt-regular/blob/34cb497affb3355f8db9bb06ae271cb4b5f10c3b/FltRegular/NumberTheory/RegularPrimes.lean#L79-L81}{\faExternalLink} and $5$\href{https://github.com/leanprover-community/flt-regular/blob/34cb497affb3355f8db9bb06ae271cb4b5f10c3b/FltRegular/FLT5.lean#L6-L8}{\faExternalLink} are regular primes (we also know that $2$ is regular\href{https://github.com/leanprover-community/flt-regular/blob/34cb497affb3355f8db9bb06ae271cb4b5f10c3b/FltRegular/NumberTheory/RegularPrimes.lean#L72}{\faExternalLink}) and the case $n=5$ of Fermat's Last Theorem follows\href{https://github.com/leanprover-community/flt-regular/blob/34cb497affb3355f8db9bb06ae271cb4b5f10c3b/FltRegular/FLT5.lean#L13}{\faExternalLink}. This is to our knowledge the first formalization of the nonexistence of nontrivial solutions to
\[
x^5 + y^5 = z ^ 5
\]
A natural question is to give more explicit examples of regular primes. One can prove that $\Z[\e^{\frac{2\pi i}{p}}]$ is principal for all $p \leq 19$ (the converse also holds), but the proof is more and more involved: indeed Minkowski's bound becomes exponentially larger and one has to check a lot of cases by hand, but more recently we have completed the cases $7$\href{https://github.com/leanprover-community/flt-regular/blob/seven/FltRegular/FLT7.lean#L12}{\faExternalLink}, $11$\href{https://github.com/leanprover-community/flt-regular/blob/seven/FltRegular/FLT11.lean#L12}{\faExternalLink} and $13$\href{https://github.com/leanprover-community/flt-regular/blob/seven/FltRegular/FLT13.lean#L12}{\faExternalLink}. Note that this is very cumbersome and it cannot work for $p > 19$ (for example $\Z[\e^{\frac{2\pi i}{23}}]$ has class number $3$, so $23$ is regular). A better approach is to prove regularity via the following result of Kummer \cite{kummer2} (see also~\cite[Theorem~5.34]{washington})
\begin{theorem} \label{thm: bernouilli}
An odd prime $p$ is regular if and only if $p$ does not divide the denominator of the Bernoulli number $B_n$ for $n = 2, 4, \ldots, p-3$.
\end{theorem}
Since Bernoulli numbers are very easy to compute (and already in \mathlib{}\href{https://github.com/leanprover-community/mathlib4/blob/4705244150fe57b4006d3393c5040d3d74e024b4/Mathlib/NumberTheory/Bernoulli.lean#L177-L179}{\faExternalLink}) this criterion gives a very easy way of checking whether a given $p$ is regular (for example one immediately obtains that the only irregular primes $p \leq 100$ are $37$, $59$ and $67$). One possible proof of Theorem~\ref{thm: bernouilli} uses the $p$-adic class number formula, which is proved using $p$-adic $L$-functions. The basic theory of $p$-adic $L$-functions has been formalized by Narayanan in \Leant{} (see~\cite{ashvni}), and a port to \Leanf{} is work in progress. For these reasons we believe that the formalization of Theorem~\ref{thm: bernouilli} is within reach, and we plan to accomplish it in the near future.

\bibliographystyle{amsalpha}
\bibliography{biblio}

\end{document}